\documentclass[10pt, conference, letterpaper]{IEEEtran}
\IEEEoverridecommandlockouts
\usepackage{amsmath,amsfonts}
\usepackage[linesnumbered,ruled,vlined]{algorithm2e}
\usepackage{array}
\usepackage{textcomp}
\usepackage{stfloats}
\usepackage{booktabs}
\usepackage{verbatim}
\usepackage{graphicx}
\usepackage{subcaption}
\usepackage{cite}
\usepackage{balance}
\usepackage{amsthm}
\usepackage{xcolor}
\usepackage{enumitem}
\usepackage{pifont}
\usepackage{bbm}
\usepackage{nicefrac, xfrac}

\theoremstyle{definition}

\newtheorem{lemma}{Lemma}

\newtheorem*{Pro*}{Problem}
\theoremstyle{remark}
\newtheorem{remark}{Remark}

\graphicspath{{figures/}{images/}{plots/}}

\def\BibTeX{{\rm B\kern-.05em{\sc i\kern-.025em b}\kern-.08em
    T\kern-.1667em\lower.7ex\hbox{E}\kern-.125emX}}
\begin{document}

\title{Closed-Form and Boundary Expressions for Task-Success Probability in Status-Driven Systems\\
}
\author{Jianpeng~Qi,
        Chao Liu,
        Rui Wang,
        Junyu~Dong,
        and~Yanwei~Yu
\thanks{J. Qi, C. Liu, J. Dong, Y. Yu are with the Faculty of Information Science and Engineering, Ocean University of China. email: \{qijianpeng, liuchao, dongjunyu, yuyanwei\}@ouc.edu.cn}
\thanks{R. Wang is with the School of Computer and Communication Engineering, University of Science and Technology Beijing. email: wangrui@ustb.edu.cn}
}

\maketitle

\begin{abstract}
Timely and efficient dissemination of server status is critical in compute-first networking (CFN) systems, where user tasks arrive dynamically and computing resources are limited and stochastic. In such systems, the access point (AP) plays a key role in forwarding tasks to a server based on its latest received server status. However, modeling the task-success probability suffering the factors of stochastic arrivals, limited server capacity, and bidirectional link delays. Therefore, we introduce a unified analytical framework that abstracts the AP’s forwarding rule as a single probability and models all network and waiting delays via their Laplace transforms. This approach yields a closed‐form expression for the end‐to‐end task‐success probability, together with upper and lower bounds that capture Erlang‐loss blocking, information staleness, and random uplink/downlink delays. We validate our results through one‐factor‐at‐a‐time simulations across a wide range of parameters, showing that theoretical predictions and bounds consistently enclose observed success rates. Our framework requires only two interchangeable inputs (the forwarding probability and the delay transforms), making it readily adaptable to alternative forwarding policies and delay distributions. Experiments demonstrate that our bounds are able to achieve accuracy within 1.0\% (upper bound) and 1.6\% (lower bound) of the empirical task‐success probability.
\end{abstract}

\section{Introduction}
\begin{figure*}[htbp]
    \centering
    \includegraphics[width=\linewidth]{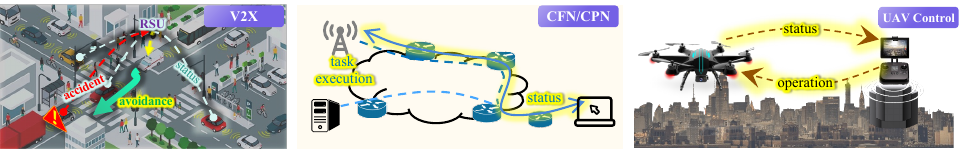}
    \caption{Scenarios of compute-aware task offloading}
    \label{fig:scenario}
    \vspace{-18pt}
\end{figure*}

In vehicular networks, vehicles act as source nodes by relaying sudden incident reports to roadside units (RSUs), which serve as destination nodes. The RSUs analyze and forecast the situation, then sending commands to nearby vehicles to facilitate safe navigation \cite{RN76}. Similarly, in Compute First Networking (CFN)\cite{RN2652, RN3, RN53}, real-time updates on computational resources are communicated to the network edge, enabling clients to accurately assess remote available resources and offload tasks accordingly. In unmanned aerial vehicle (UAV) control systems, drones transmit environmental data collected on-site to a control station, which processes the information and issues operational directives to accomplish reconnaissance tasks. As depicted in \figurename~\ref{fig:scenario}, these scenarios illustrate a structured closed-loop system consisting of four essential stages: ``Status updating'', ``status consumption'', ``task forwarding'', and ``task execution''. The cyclic process of transmitting resource or environmental data and executing subsequent tasks is widespread across various networking scenarios. Despite its growing importance, guaranteeing task completion in those status-driven systems remains challenging due to inherent uncertainties arising from stochastic task arrivals, limited resources, and random transmission delays.

From a more general perspective, terminal devices frequently offload locally generated computational tasks to remote service nodes (edge servers) for execution. To assess the availability of a service node, servers periodically or conditionally transmit their current resource status to access points (APs) when significant changes occur or when coordination is necessary. Subsequently, the AP decides whether to offload the client's task based on this information \cite{10816182, RN53}. Therefore, the timeliness and accuracy of service information (i.e, resource status) in these systems directly impact the likelihood of successful task execution. Incorrect or outdated status can cause the AP to misjudge the availability or capability of service nodes, leading to failed task redirection or inefficient offloading decisions. 

Actually, the research community has introduced related theoretical frameworks such as Age of Information (AoI) \cite{RN1372} and Value of Information (VoI) to assess and optimize information update mechanisms \cite{RN2760, 9757236}. AoI is a performance metric designed to quantify the timeliness of information. It focuses on the time elapsed between the generation of an update and its reception at the destination. 
Extended variants such as AoP (Age of Processing) \cite{RN1181} and QAoI (Query Age of Information) \cite{RN1437} have been proposed. AoP captures delays in processing, while QAoI specifically reflects the AoI at the time it is queried or consumed by the decision-making node, offering a more comprehensive view of latency and relevance in complex systems. VoI evaluates the importance of information from a semantic and functional perspective. In essence, it quantifies the degree to which a piece of information enhances decision making or the utility of the system. For example, in network control systems, the ``value'' of the information can be interpreted as the amount by which it reduces environmental uncertainty \cite{RN2761}. 

However, there has been little research establishing the relationship between \textit{task-success probability} and the various dynamic factors in status-driven closed-loop systems. Exploring this relationship will provide valuable insights into how system performance responds to changing conditions.
Therefore, in this paper, we investigate the problem of modeling and analyzing task-success probability in status-driven systems. To lower the level of difficulty, we mainly focus on a \textit{single-server, single-destination} scenario, aiming to build the closed-form expression of the achievable task-success probability under dynamic task arrivals, stochastic AP forwarding decisions, multi-thread server capabilities, and stochastic uplink and downlink delays. 

We first express the AP’s forwarding rule as a single forwarding probability, the chance that a newly arrived task will actually be sent to the server. We then apply a fixed-point iteration to determine the effective forwarding rate that corresponds to this probability. Next, we derive a closed-form expression for the task-success probability, and the upper bound based on the Erlang-B formula.
However, deriving the lower bound is not straightforward. To address this, we examine the critical state where exactly one thread remains idle and analyze how frequently the system transitions from ``one free'' to ``none free,'' which determines the risk of resource exhaustion.
We then incorporate the three independent delays (the uplink delay, the waiting age, and the downlink delay) using their Laplace transforms, multiply the relevant factors, and obtain the concise lower bounds. 

Our experiments show that the closed-form theoretical probability almost perfectly fit the observed task-success results. For instance, when varying the task-arrival rate, the upper bound is within 1\% and the lower bound within 1.6\% of the theoretical value, effectively enclosing the empirical results. Surprising, we also find that increasing the status update frequency brings only modest improvements in the task-success probability.
Because our methodology requires only these two inputs (the forwarding probability and the Laplace transform of the delay) it is straightforward to adapt to different admission rules, link conditions, or arrival patterns simply by updating these values, while leaving the rest of the framework unchanged.

This paper makes the following key contributions:
\begin{itemize}
    \item We propose an unified analytical framework that jointly considers resource blocking, AoI-based status staleness, and stochastic uplink and downlink delays. Based on that, we derive the closed-form of task-success probability. 
    \item To facilitate both analysis and practical implementation, by leveraging Laplace transforms and renewal process theory, we derive upper and lower bounds.
    \item We evaluate how well the closed-form expression and analytical bounds capture the system’s behavior. Theoretical results closely match the experimental outcomes, and our experiments further validate the proposed theories on resource consumption and user arrivals.
\end{itemize}

\section{Related Work} \label{sect:related-work}

Ensuring a high task-success probability is a central challenge in status-driven CFN systems. The accuracy and timeliness of resource status updates directly influence the likelihood that tasks will successfully execute. The AoI metric, which quantifies the freshness of status updates, is commonly employed to evaluate this relationship \cite{RN1372}. Over the past few years, numerous works have extended the AoI concept to better capture system nuances \cite{RN39, RN48, RN83}. 
Variants such as AoP \cite{RN1181}, Age of Event (AoE) \cite{10.1145/3717834}, TPAoI (Three-Phase AoI) \cite{RN2759} and remote inference \cite{RN51}. They aim to better represent real-world system dynamics and their impact on task results. For example, the TPAoI explicitly captures update generation, status consumption, and task dispatch stages, linking information freshness directly to task execution outcomes.

Nevertheless, the stochastic system states complicates the accurate prediction of resource availability solely through AoI-based metrics \cite{RN51, RN34, RN14}. User interactions further amplify this complexity by introducing dynamics that traditional AoI metrics cannot fully capture, potentially leading to incorrect task-forwarding decisions and reduced success rates. For instance, QAoI \cite{RN1437} extends AoI to measure age precisely at user-query instants and has been evaluated under both deterministic and uniformly distributed query arrivals. To overcome these limitations, researchers have proposed advanced metrics such as VoI, which emphasizes the utility of updates, addressing both freshness and semantic relevance to improve task-success probability \cite{RN2760, RN2761}.

Recent studies have increasingly recognized the necessity of integrating client behavior and task execution dynamics into metrics evaluating information effectiveness \cite{RN40, RN1181, RN26, RN37, RN38, RN1437}. 
Despite these advancements, explicitly quantifying and modeling the relationship between information freshness, resource status accuracy, and task-success probability remains an open research problem in CFN.

To introduce determinism into system performance, several analytical studies have established closed‐form and bounded expressions for AoI and its variants under diverse delay models and system settings. Based on a Rayleigh block fading model and assuming each transmission consumes one time slot, \cite{RN1198} derives closed-form equations for the probability mass function (PMF) of AoI in an $N$-hop network with time-invariant packet loss probabilities on each link. \cite{RN48} derives the average AoI in a two‐way delay system assuming geometrically distributed link latencies. \cite{RN37} analyzes optimal update‐rate bounds when the uplink delay distribution is unknown. \cite{RN34} examines the Age of Incorrect Information, identifying threshold‐based update policies for geometric or Zipf‐distributed transmission delays. In multi‐source scenarios with correlated updates, \cite{10.1145/3492866.3549719} characterizes weighted‐sum average AoI bounds under channel‐conflict conditions. 
\cite{9757236} defines the VoI for hidden Markov models as the mutual information between current and past dynamic observations. Specializing to the Ornstein–Uhlenbeck process, they derive a closed-form expression for VoI and evaluate it under various sampling policies (arbitrary fixed, uniform, and Poisson) as well as for exponentially distributed transmission delays. 
More broadly, \cite{RN1372} provide boundary analyses for AoI in single‐ and multi‐server queuing systems, including cases with exponentially distributed delays. To retain analytical tractability in our framework, we likewise assume exponential uplink and downlink latencies.  In addition, numerous works employ the Lyapunov optimization \cite{RN1618}  or reinforcement learning (both model-based methods such as \cite{RN1437} and model-free approaches like \cite{RN25}) to optimize update policies and indirectly enhance system performance; since our focus is on the direct derivation and analysis of task-success probability, we omit further discussion of these methods.

In summary, recent research converges on the importance of status updating and status-driven task offloading: Systems like CFN require up-to-date and relevant status information to make correct offloading decisions, and theoretical frameworks like AoI/VoI have been extended to guide when and what to update. Building on these insights, our study seeks to explicitly relate task success probability to factors such as resource uncertainty, update staleness, and network delays, which have been less examined in previous research.

\section{System Model and Problem Formulation} \label{sect:system-model}
In this section, we first describe the system model, then formalize its underlying dynamics, and finally state the problem addressed in this paper.

\subsection{System Model}
\figurename~\ref{fig:system-model} depicts a single-source, single-destination CFN scenario: A single AP communicates with a remote server that can host at most $C$ concurrent threads (or virtual containers), each thread with processing speed $\mu$. The server sends status packets periodically with rate $\bar r$ to the AP over an \emph{uplink} with transmission rate~$\gamma$ and receives user-task packets from the AP over a \emph{downlink} with rate~$\beta$. The original user task arrival process follows Poisson process with parameter $\lambda_{in}$ \cite{RN53}. The server’s instantaneous availability is the scalar $c(t)$, the number of idle threads at time~$t$. Whenever a new status packet arrives, the AP records this value as $\hat c(t)$; the cache remains unchanged until the next update, even though $c(t)$ may fluctuate meanwhile. 

\begin{figure}[htbp]
    \centering
    \includegraphics[width=\linewidth]{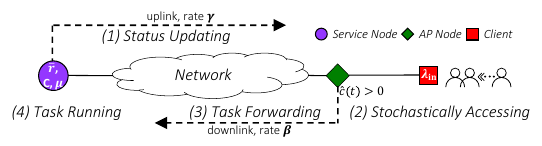}
    \caption{Illustration of the information-consumption stages}
    \label{fig:system-model}
\end{figure}

When serving, the overall process of this system consists of the following four stages: 
\begin{enumerate}
\item \textit{Status Updating}: The server transmits its current availability $c(t)$ to the AP over the uplink.
\item \textit{Stochastically Accessing}: User tasks reach the AP according to a Poisson process of rate $\lambda_{\mathrm{in}}$.
\item \textit{Task forwarding}: Upon each task arrival, the AP checks $\hat c(t)$. If $\hat c(t)>0$, the task is forwarded immediately; otherwise it is dropped at the AP.
\item \textit{Task Running}: A forwarded task occupies one thread and is processed with an service time of $1/\mu$. If all $C$ threads are busy on arrival, the task is discarded (no buffer). After completion, the server returns the result to the client; the return latency is beyond the scope of this study.
\end{enumerate}

\begin{remark}
The instantaneous server state is captured by a single scalar $c(t)$, interpreted as the currently available compute capacity. A larger value of $c(t)$ therefore denotes stronger capability. When the server exposes heterogeneous resources (CPU, memory, GPU slots, network bandwidth, etc.), a direct scalar may be inadequate; in that case one can learn an embedding function $c(t)=\phi(\mathbf{r}(t))$ that maps the multi-dimensional resource and demand vector $\mathbf{r}(t)$ to a one-dimensional score. The mapping $\phi(\cdot)$ can be instantiated by a lightweight deep neural network trained offline (or periodically fine-tuned) to minimize task-completion violations under representative workloads. 
\end{remark}

Some important symbols and explanations can be found in \tablename~\ref{tbl:symbols}. 
\begin{table*}[htbp]
    \caption{List of important symbols}
    \label{tbl:symbols}
    \centering
    \renewcommand{\arraystretch}{1.2}
    \begin{tabular}{c|p{6.5cm}|c|p{6.5cm}}
    \toprule
    \textbf{Symbol} & \textbf{Meaning / Default Assumption} 
    & \textbf{Symbol} & \textbf{Meaning / Default Assumption} \\ 
    \hline 
    $\lambda_{\mathrm{in}}$ 
      & Exogenous user-task arrival rate at the AP. 
      & $\lambda$ 
      & AP forwarding rate after thinning; $\lambda\le\lambda_{\mathrm{in}}$.\\ 
    $C$ 
      & Maximum number of concurrent server threads. 
      & $c(t)$ 
      & Instantaneous number of idle threads; $0 \le c(t) \le C$. \\[2pt]
    $\hat c(t)$ 
      & Most recent cached value of $c(t)$ stored at the AP. 
      & $\mathbb E[Y]$ 
      & Mean inter-update interval; $\mathbb E[Y] = 1/\bar r$, where $\bar r$ is the average update frequency (controllable). \\[2pt]
    $Y_k$ 
      & $k$th inter-update gap at AP.
      & $S_k$ 
      & Arrival time of the $k$th update, 
      ($S_k=\sum_{i<k}Y_i$).
      \\[2pt]
    $N(t)$ 
      & Number of updates received by time $t$: $N(t)=\max\{k:S_k\le t\}$. 
      & $A(t)$ 
      & Elapsed time since the $N(t)$-th update at the AP: $A(t) = t - S_{N(t)}$. \\[2pt]
    $\Delta$ 
      & $D+A$: AoI from status generation to AP consumption. 
      & $H$ 
      & $\Delta+F$: total decision window from state generation to task arrival at the server. \\[2pt]
    $\Pr\{\hat c(t)>0\}$ 
      & AP's forwarding probability. 
      & $P_{\mathrm{idle}}$ 
      & Probability that the server has at least one idle thread.\\
    $\gamma$ 
      & Uplink transmission rate for status updates. 
      & $D$ 
      & Uplink delay of a status packet, $\mathrm{Exp}(\gamma)$. \\[2pt]
    $\beta$ 
      & Downlink transmission rate for task packets. 
      & $F$ 
      & Downlink delay of a task packet, $\mathrm{Exp}(\beta)$. \\[2pt]
    $\mu$ 
      & Per-thread service rate at the server. 
      & $\rho$ 
      & Offered load (server utilization) $\lambda/\mu$. \\[2pt]
    $\Lambda$ 
      & Resource-exhaustion hazard rate: $\Lambda = \lambda\bigl[1 - B(\rho,C)\bigr]$. 
      & $S(x)$ 
      & Probability that the server remains non-empty in $[0,x]$.\\ 
      \bottomrule
    \end{tabular}
    \vspace{-18pt}
\end{table*}
\subsection{System Dynamics and Problem Formulation}
The system model shows that several stochastic processes interact within the CFN scenario. To structure our analysis, we now formalize these processes from two viewpoints: (i) The \emph{task side}, which tracks how user requests arrive, are forwarded by the AP, and compete for server threads, and (ii) the \emph{status side}, which describes how the server’s status information are generated, transmitted, cached, and consumed.

\textbf{\textit{Task-related processes.}}
The \emph{Server} operates as an pure-loss (Erlang-loss) system: $C$ execution threads
work in parallel and each thread completes its task after an 
service time with $1/\mu$.
Recall that user tasks reach the AP according to an exogenous Poisson process with rate $\lambda_{in}$. Upon each arrival the AP consults $\hat c(t)$:
\begin{itemize}
    \item If $\hat{c}(t) > 0$, the task is forwarded immediately;
    \item If $\hat{c}(t) = 0$, the task is discarded on the spot.
\end{itemize}

The rate of forwarding the task to the server is denoted as 
\begin{equation}
    \lambda =\lambda_{\mathrm{in}}
           \mathbbm 1_{\{\hat c(t)>0\}},
\end{equation}
where $\mathbbm{1}_{\{\hat c(t)>0\}}$ is the indicator that the cached state shows at least one idle thread.
While idle threads are available, tasks arrive according to a Poisson process with rate $\lambda$; otherwise they are blocked.

The steady-state probability that all $C$ execution threads are simultaneously \emph{busy} (i.e., the server is fully occupied) is given by the classical Erlang-B formula
\begin{equation}
    B(\rho,C)=\frac{
        \rho^C/C!
    }{
        \sum_{k=0}^{C}\rho^k/k!
    },
\end{equation}
where the ratio $\rho=\lambda/\mu$ is the offered load (utilization).
Consequently, the steady-state probability that the server retains at least one idle thread (hence remains capable of accepting a newly forwarded task) is
\begin{equation} \label{eq:p-idle}
    P_{idle}=1-B(\rho,C).
\end{equation}

\textbf{\textit{Status updating processes.}} 
Status packets arrive at the AP at time $S_0=0, ~ S_k = \sum_{i=0}^{k-1}Y_i, ~ k\geq 1$, where the inter-update gaps (at AP) $\{Y_k\}_{k\geq 0}$ are independent and identically distributed.

Let $N(t)=max\{ k: S_k \leq t \}$ be the index of the most recent packet by time $t$. The \emph{age process} 
\begin{equation}
    A(t)=t - S_{N(t)}
\end{equation}
measures how long the AP has been relying on a possibly stale snapshot.

Following the approach adopted in previous studies \cite{9757236,RN1372}, we model the delays in both uplink (denoted as $D$) and downlink (or forwarding link, denoted as $f$) as independent exponential random variables with parameters $\gamma$ and $\beta$, respectively. It is important to emphasize that the choice of this particular distribution does not limit the general applicability of our analytical framework. 

Coupled with the uplink delay $D$ and the downlink delay $F$, the total time lag between the generation of a status packet and the moment the corresponding task touches the server is 
\begin{equation}\label{eq:effective_window}
    H=D+A(t)+F.
\end{equation}
We assume all three components are mutually independent. This lag, hereafter called the \emph{holding window}, determines whether the server still owns resources when the task actually arrives.  
Note that, the definition of $A(t)$ is slightly different from the classical AoI $\Delta$ \cite{RN1372}, which also contains the time of transmitting (i.e., $\Delta = D+A(t)$).

\begin{Pro*}[Task-Success Probability $P_\text{succ}\left(\bar{r},C,\mu,\lambda_{in},\gamma,\beta\right)$]
Given the status update rate $\bar{r}$, uplink transmission rate $\gamma$, downlink transmission rate $\beta$, task arrival rate $\lambda_{\mathrm{in}}$, per-thread service rate $\mu$, and the maximum number of concurrent server threads $C$, let $N_{\text{succ}}$ denote the number of tasks that successfully complete task on the server and $N_{\text{arr}}$ the total number of tasks arriving at the AP. The empirical task-success probability is defined as $\hat P_{\text{succ}} = \frac{N_{\text{succ}}}{N_{\text{arr}}}$.
Our goal is to derive a theoretical closed-form expression for the achievable task-success probability $P_\text{succ}(\bar{r},\gamma,\beta,\lambda_{\mathrm{in}},\mu,C)$ that closely matches real system behavior, as well as analytical upper and lower boundaries for $P_\text{succ}$ under the given system parameters.
\end{Pro*}

As described earlier in the system model, the solution to this problem jointly captures the effects of dynamic server capacity, user access behavior, information staleness, and bidirectional random link delays. In doing so, it provides a unified and precise foundation for the stochastic analysis and optimization that follow.

\section{Dynamic Analyses}\label{sect:dynamics}
\subsection{AP's Forwarding Probability $P_{\mathrm{fwd}}$ and Rate $\lambda$}
Not all tasks arriving at the AP are forwarded to the server, so it is crucial to derive both the probability and rate at which tasks are actually admitted. In this section, we present the forwarding rate $\lambda$ at the AP and provide its solution.

\begin{lemma}\label{lem:forwarding-prob}
If the arrival mechanism of status-update packets is statistically independent of the server’s queuing process, then the AP forwards an incoming task with probability 
\begin{align}
    P_{\mathrm{fwd}} :=& Pr\{\hat{c}(t)>0\} = P_{idle},
\end{align}
and the resulting forwarding rate is 
\begin{equation} \label{eq:lambda}
    \lambda = \lambda_{in}P_{idle}.
\end{equation}
\end{lemma}

\begin{proof}
For each update cycle $k$ set
\begin{equation}
   R_k := \mathbbm 1_{\{c(S_k^-)>0\}},
\end{equation}
where the superscript ``$-$'' means the left limit, so the packet that
arrives at $S_k$ itself is not yet taken into account.  
Thus $R_k=1$ \textit{iff} the server is \emph{not} fully occupied just before
the cache is refreshed.

Immediately after a packet arrives the AP rewrites its cache,
$\hat c(S_k)=c(S_k^-)$, and keeps that value unchanged for the whole
interval $[S_k,S_{k+1})$.
Hence, for every $t\in[S_k,S_{k+1})$,
\begin{align}
   I(t) := & \mathbbm 1_{\{\hat c(t)>0\}} = R_k .
\end{align}

The ``reward'' earned in cycle $k$ is the total time
during which the AP believes forwarding is allowed:
\begin{align}
   Q_k :=& \int_{S_k}^{S_{k+1}} I(t)\,\mathrm dt = R_k\,Y_k . 
\end{align}
For any horizon $T>0$ choose $n(T):=\max\{k:S_k<T\}$.
Partition $[0,T]$ into complete cycles and the final stub:
\begin{equation}
   \frac1T\int_{0}^{T} I(t)\,\mathrm dt
   =\frac{1}{T}\sum_{k=0}^{n(T)-1} Q_k
    \;+\;
    \frac{1}{T}\int_{S_{n(T)}}^{T} I(t)\,\mathrm dt .
\end{equation}
Because $I(t) \in \{0, 1\}$, the stub term $\tfrac{1}{T}\int_{S_{n(T)}}^{T} I(t)$ is bounded above by
$\tfrac{(T-S_{n(T)})}{T}\le \tfrac{Y_{n(T)}}{T}\to0$ almost surely as $T\to\infty$.
Consequently,
\begin{equation}
   \lim_{T\to\infty}\frac1T\int_{0}^{T} I(t)\,\mathrm dt
   = \lim_{T\to\infty}
     \frac{1}{\sum_{k=0}^{n(T)-1} Y_k}\sum_{k=0}^{n(T)-1} Q_k .
\end{equation}

Since the cycles $(Y_k,Q_k)$ are i.i.d., through the renewal reward theorem \cite{ross1995stochastic}, the long-run average exists
and equals the ratio of expectations:
\begin{align}
   P_{\mathrm{fwd}}
                    = \frac{\mathbb E[Q_k]}{\mathbb E[Y_k]} . 
\end{align}

Because the update sampling mechanism is independent of the content $c(t)$, $R_k$ and $Y_k$ are uncorrelated, i.e., $\operatorname{Cov}(R_k,Y_k)=0$.
Hence
\begin{equation}
    \mathbb E[Q_k]
   =\mathbb E[R_k\,Y_k]
   =\mathbb E[R_k]\,\mathbb E[Y_k].
\end{equation}
Therefore, the forwarding probability is
\begin{align}
   P_{\mathrm{fwd}} =\mathbb E[R_k] =\Pr\{c(S_k^-)>0\} =P_{\mathrm{idle}} . 
\end{align}

User tasks arrive at the AP as a Poisson flow of rate
$\lambda_{\mathrm{in}}$ and are independent of $\hat c(t)$. Thinning a Poisson process therefore yields a forwarded Poisson flow of rate
\begin{equation}
  \lambda = \lambda_{\mathrm{in}}\,P_{\mathrm{idle}}.
\end{equation}
\qedhere
\end{proof}

\textbf{\emph{Solving \eqref{eq:lambda}.}}
Combining the two relations, \eqref{eq:p-idle} and \eqref{eq:lambda}, and
eliminates $P_{\mathrm{idle}}$ and yields a single equation for
$\lambda$:
\begin{equation}\label{eq:fp}
  \lambda
  \;=\;
  \lambda_{\mathrm{in}}
  \Bigl[1-B\!\bigl(\tfrac{\lambda}{\mu},\,C\bigr)\Bigr].
\end{equation}
We can solve this using fixed-point iteration.

\noindent \textit{An simple explanation for using fixed-point iteration: }
Let
$g(\lambda):=\lambda_{\mathrm{in}}\bigl[1-B(\lambda/\mu,\,C)\bigr]$
denote the right-hand side of~\eqref{eq:fp}.  
Observe two simple facts:
\begin{itemize}
    \item $B(\cdot,C)$ is strictly increasing, so $1-B(\cdot,C)$ is strictly decreasing; hence $g(\lambda)$ is strictly decreasing in $\lambda$.
  \item $g(0)=\lambda_{\mathrm{in}}$ (above the line $y=\lambda$),  
  whereas  
  $g(\lambda_{\mathrm{in}})=\lambda_{\mathrm{in}}\bigl[1-
  B(\lambda_{\mathrm{in}}/\mu,\,C)\bigr]\;<\;\lambda_{\mathrm{in}}$
  (below the line $y=\lambda$).
\end{itemize}
Therefore, the two curves $y=\lambda$ and $y=g(\lambda)$ cross \emph{exactly once} on the interval $[0,\lambda_{\mathrm{in}}]$.  
Denote this unique intersection by $\lambda^{\ast}$, and finally we have $\lambda^{\ast} = g(\lambda^{\ast})$.

Therefore, in practice, a one-line fixed-point iteration converges monotonically to
$\lambda^{\ast}$:
\begin{align}
  \lambda^{(k+1)} \;=\; 
  \lambda_{\mathrm{in}}\Bigl[1-
  B\!\bigl(\tfrac{\lambda^{(k)}}{\mu},\,C\bigr)\Bigr],
  \qquad
  \lambda^{(0)} = \lambda_{\mathrm{in}} .
\end{align}
Starting from $\lambda^{(0)}$ (or any value in
$[0,\lambda_{\mathrm{in}}]$), the sequence
$\{\lambda^{(k)}\}$ moves downward and converges to the unique
fixed point $\lambda^{\ast}$. Algorithm \ref{alg:solve-lambda} shows this procedure. We set the max-iteration $I_{max}=500$ and $\varepsilon = 10^{-12}$ as default values in our experiments.
\begin{algorithm}[htbp]
\caption{Fixed-point iteration for solving $\lambda$}
\label{alg:solve-lambda}
\KwIn{exogenous rate $\lambda_{\mathrm{in}}$, service rate $\mu$,
      thread count $C$, tolerance $\varepsilon$, max‐iteration $I_{\max}$}
\KwOut{forwarding rate $\lambda$}

\BlankLine
$\lambda^{(0)} \leftarrow \lambda_{\mathrm{in}}$\;
\For{$k \leftarrow 0$ \KwTo $I_{\max}-1$}{
      $\rho \leftarrow \lambda^{(k)} / \mu$ \;
      $\lambda^{(k+1)} \leftarrow \lambda_{\mathrm{in}}\,(1-B(\rho,C))$ \;
      \If{$\lvert\lambda^{(k+1)} -\lambda^{(k)} \rvert < \varepsilon$}{
            \Return{$\lambda^{(k+1)}$}\;
      }
  }
\end{algorithm}

\subsection{Theoretical Success Probability and Resource Hazard Rate}\label{sec:hazard_rate}
Lemma \ref{lem:forwarding-prob} specifies only the \emph{arrival rate} of forwarded tasks; it does not guarantee that every forwarded task will eventually be executed.
In practice, a task’s success also depends on the server’s instantaneous resource availability. 

\begin{lemma}\label{lem:theoretical_p}
Under a Poisson arrival process of rate \(\lambda_{\mathrm{in}}\) and a server governed by the Erlang–B loss model, the theoretical task–success probability of the system admits the closed‐form
\begin{equation}
  P_{\mathrm{succ}}
  =P_{\mathrm{fwd}}\times P_{\mathrm{idle}}
  =P_{\mathrm{idle}}\times P_{\mathrm{idle}}\,.
\end{equation}
\end{lemma}

The result of Lemma \ref{lem:theoretical_p} is immediate from the fact that a task is forwarded if and only if the server has at least one idle thread (probability \(P_{\mathrm{idle}}\)), and a forwarded task succeeds precisely when the server remains idle at the instant of admission (again probability \(P_{\mathrm{idle}}\)); hence the proof is omitted.  

Although Lemma \ref{lem:theoretical_p} offers valuable theoretical insight, it provides no guarantee on the stochastic behavior of the system. Therefore, given the task-arrival rate, we proceed to derive a rigorous lower bound on the task–success probability $P_\text{succ}$. 
At first, we  derive the server resource-exhaustion process.
\begin{lemma}\label{lem:hazard}
Let 
\begin{equation}
  P_1 \;:=\;\Pr\{c(t)=1\}
        \;=\;\frac{\rho^{\,C-1}/(C-1)!}{\sum_{k=0}^C\rho^k/k!}
        \;=\;B\bigl(\rho,C-1\bigr)
\end{equation}
be the steady‐state probability that exactly one thread is idle.  Then the hazard rate at which the server becomes fully occupied (i.e.\ transitions from one idle thread to zero) is
\begin{equation}
  \Lambda \;=\;\lambda \,P_1.
\end{equation}
\end{lemma}

\begin{proof}
We wish to compute the instantaneous \emph{hazard rate} $\Lambda$ at which the process $c(t)$ jumps from state $1$ to state $0$.  Since $c(t)\in\{0,1,\dots,C\}$ evolves by single‐step births (service completions) and deaths (task admissions), the only way to go from $1$ to $0$ is for a task admission to occur when $c(t^-)=1$.

By construction, the AP admits (forwards) user tasks as a Poisson process of rate $\lambda$.
By PASTA, an admitted task observes the server occupancy in steady state.  Hence the rate at which admissions find \(c(t)=1\) is
\[
  \lambda \;\times\;\Pr\{c(t)=1\}
  \;=\;\lambda\,P_1.
\]
Such an admission necessarily drives $c(t)$ from $1$ down to $0$, since each admitted task occupies one idle thread.  Therefore this product is exactly the hazard rate of resource exhaustion:
\[
  \Lambda
  = \lambda\,P_1.
\]
\end{proof}

Lemma \ref{lem:hazard} means a larger blocking probability (fewer threads or heavier load) implies a smaller $P_{\mathrm{idle}}$ and therefore a smaller $\Lambda$: With less spare capacity, fewer incoming tasks can actually seize the last idle thread.

\subsection{Holding-Window Survival Probability}
After deriving the hazard rate $\Lambda$, we further leverage it to compute the probability that the resource pool remains non-exhausted throughout a given holding window $H$, i.e., the holding-window survival probability. 

Define the stopping time
$T:=\inf\{t>0 : c(t)=0\}$, i.e.\ the next instant at which all
threads become busy.  
Let
\begin{equation}
   S(x):=\Pr\{T>x\},
\end{equation}
the \emph{survival probability} that \emph{no} resource–exhaustion occurs in the entire
interval $[0,x]$.

\begin{lemma}\label{lem:survival}
With hazard rate $\Lambda=\lambda P_1$,
\begin{equation}\label{eq:s_x}
   S(x)=e^{-\Lambda x}, \qquad x\ge0 .
\end{equation}
\end{lemma}

\begin{proof}
Treat the exhaustion events as a Poisson process
$M_\Lambda(t)$ with rate $\Lambda$.
Then $T>x$ is equivalent to ``zero events occur in $[0,x]$'', i.e., $M_\Lambda(x)=0$.
Therefore
\begin{equation}
   S(x)=\Pr\{M_\Lambda(x)=0\}
        =e^{-\Lambda x}.
\end{equation}
\end{proof}

Eq. \eqref{eq:s_x} quantifies the chance that \emph{every} moment within a window of length $x$ will find at least one idle thread.

Lemma \ref{lem:survival} shows that if the holding window has a
\emph{fixed} length~$x$, the probability of keeping at least one idle
thread for the whole interval equals
$S(x)=e^{-\Lambda x}$.
In reality, deciding the window length $x=H$ in \eqref{eq:effective_window} is non-trivial, because the dynamics existed in $D$, $A(t)$, and $F$.

\begin{lemma}\label{lem:SH-final}
Under the usual independence assumptions (i) the update sequence $\{Y_k\}$ is independent of the user–task arrival process of rate $\lambda_{\mathrm{in}}$;
(ii) the AP forwards a task \emph{iff} $\hat c(t)>0$, the survival probability for the holding window $H$ is
\begin{equation}\label{eq:SH-final}
   S(H)=
   \frac{1-L_Y(\Lambda)}{\Lambda\,\mathbb E[Y]}
   \;\frac{\gamma}{\Lambda+\gamma}
   \;\frac{\beta}{\Lambda+\beta},
\end{equation}
where $L_Y(\Lambda)=\int_{0}^{\infty}{e^{-\Lambda y}dF_Y\left(y\right)}=\mathbb E[e^{-\Lambda Y}]$ is the
Laplace–Stieltjes transform (LST) of the inter-update gap~$Y$
evaluated at $s=\Lambda$.
\end{lemma}

\begin{proof}
Since $D$ and $F$ are exponential, $f_D(d)=\gamma e^{-\gamma d}$ and $f_F(f)=\beta e^{-\beta f}$.
Because $A(t)$ is the \emph{stationary age} (residual life) of the renewal
process $\{Y_k\}$, its tail probability is obtained by the
well-known renewal–reward identity
\begin{equation}
   \Pr\{A(t)>a\}
   =\frac{\displaystyle\int_{a}^{\infty}\Pr\{Y>u\}\,du}
          {\mathbb E[Y]}
   =\frac{\displaystyle\int_{a}^{\infty}\bigl[1-F_Y(u)\bigr]\,du}
          {\mathbb E[Y]},
\end{equation}
where $a\ge0$. Differentiating with respect to $a$ yields the pdf of $A(t)$:
\begin{equation}
   f_A(a)
   = -\frac{d}{da}\Pr\{A>a\}
   = \frac{1-F_Y(a)}{\mathbb E[Y]}
   = \frac{\Pr\{Y>a\}}{\mathbb E[Y]}.
\end{equation}

Because $D$, $A(t)$ and $F$ are independent, their sum has
density
$f_H = f_D * f_A * f_F$, i.e.\
\begin{equation}
   f_H(h)=\int_{0}^{h}
           \!\!\int_{0}^{u}
           f_D(d)\,f_A(u-d)\,f_F(h-u)\,dd\,du .
\end{equation}

For a realised window length $h$,
$S(h)=e^{-\Lambda h}$.
Averaging over the distribution of $H$ gives
\begin{equation}\label{eq:conv}
   S(H)=\int_{0}^{\infty} e^{-\Lambda h}\,f_H(h)\,dh .
\end{equation}

Eq. \eqref{eq:conv} involves multiple convolutions, making it difficult to solve. 

For any non–negative random variable $X$, its LST is defined on the real half-line as
\begin{equation}
    \mathcal{L}_{X}(s)=\mathbb{E}\!\bigl[e^{-sX}\bigr], 
   \qquad s\ge 0 .
\end{equation}
Since $\Lambda \ge 0$, therefore, we introduce the LST; taking $s=\Lambda$ yields
\begin{equation}
   \mathcal L_H(\Lambda)=
   \frac{1-L_Y(\Lambda)}{\Lambda\,\mathbb E[Y]}\;
   \frac{\gamma}{\Lambda+\gamma}\;
   \frac{\beta}{\Lambda+\beta}.
\end{equation}
Because $\mathcal L_H(\Lambda)$ is exactly the integral, \eqref{eq:SH-final} follows.
\end{proof}

The factor $\tfrac{1-L_Y(\Lambda)}{\Lambda\mathbb E[Y]}$
captures the impact of \emph{information staleness};
the other two factors quantify attenuation due to
\emph{uplink} and \emph{downlink} delays.
Since $A(t)$ is already the stationary age of the renewal process and
does not interact with queueing or forwarding, \emph{explicitly}
including or omitting the waiting interval~$A$ leaves the
final success-probability formula unchanged;
all waiting effects are absorbed by $G_Y(\Lambda)$.

\begin{remark}
Our methodology may also adaptable for other distributions of uplink and downlink delays.
Meanwhile, to port all continuous-time derivations into a slotted system with slot width $u$, multiply every ``per-second'' rate
$(\lambda, \mu, \gamma, \beta, \Lambda)$ by $u$ to obtain the
corresponding ``per-slot'' probabilities $(\alpha, q, \gamma_d, \beta_d, \varphi)$.
Then make the following one-for-one substitutions throughout:
$
   e^{-\Lambda x}\;\mapsto\;(1-\varphi)^{x},
$, 
$
   L\text{aplace transform}\;\mapsto\;z\text{-transform}
$, and
$
   \int \; \mapsto\; \sum .
$
With these replacements the entire analytical framework, carries over unchanged to discrete time.
\end{remark}

\section{Two Boundaries}

For brevity we henceforth write $P_{\mathrm{succ}}$ in the compact form
\begin{equation}\label{eq:p_success}
   P_{\text{succ}} \;=\; \bigl(1- B(\rho,C)\bigr)\,
   G_Y(\Lambda)\,
   G_\gamma(\Lambda)\,
   G_\beta(\Lambda),
\end{equation}
where
\[
   G_Y(\Lambda)=\frac{1-L_Y(\Lambda)}{\Lambda\mathbb E[Y]}, 
   G_\gamma(\Lambda)  = \frac{\gamma}{\Lambda+\gamma}, 
   G_\beta(\Lambda)   = \frac{\beta}{\Lambda+\beta}.
\]

By \eqref{eq:p_success}, we can derive:
\begin{enumerate}
\item 
      If both link rates tend to infinity
      $(\gamma,\beta\to\infty)$, their delays vanish
      $(D,F\to0)$, then $G_\gamma(\Lambda)$ and $G_\beta(\Lambda)$ tend to $1$. Therefore, \eqref{eq:Psucc} reduces to
      \begin{equation}\label{eq:no_gamma_beta}
        P_{\text{succ}}
        \;\longrightarrow\;
        \bigl[1-B(\rho,C)\bigr]\;
        \frac{1-L_Y(\Lambda)}{\Lambda\,\mathbb E[Y]} .
      \end{equation}
\item 
      Taking \ref{eq:no_gamma_beta} and further letting the update rate
      $\bar r=1/\mathbb E[Y]$ grow without bound drives the information staleness
      factor to one, so that
      \begin{equation}
        P_{\text{succ}}
        \;\longrightarrow\;
        1-B(\rho,C),
      \end{equation}
      i.e.,\ the ideal Erlang-loss upper bound in which the AP
      possesses perfectly fresh information.
\end{enumerate}

\emph{\textbf{Ideal upper bound} $u_{ideal}$.}
When both links have infinite bandwidth ($\gamma,\beta\!\to\!\infty$) and
the update rate tends to infinity ($\bar r\!\to\!\infty$),
$D\!=\!F\!=\!0$ and $G_Y\!=\!1$,
so the success probability is ideally bounded by
\begin{equation}
    u_{idle} \; = \; 1-B(\rho,C).
\end{equation}

\emph{\textbf{Operational lower bound} $l_{opt}$.} 
Combining the AP forwarding probability $P_{\mathrm{idle}}$ with the holding window survival probability $S(H)$ in \eqref{eq:SH-final} yields the following lower-bound task–success probability:
\begin{equation}\label{eq:Psucc}
   P_{\text{succ}}
   \;=\;
   P_{\mathrm{idle}}\;S(H).
\end{equation}

Recall that, by LST,
\begin{equation}
   G_Y(\Lambda)\,G_\gamma(\Lambda)\,G_\beta(\Lambda)
   =\mathbb{E}\!\bigl[e^{-\Lambda(\Delta+F)}\bigr],
\end{equation}
where $\Delta=D+A$.
It lets us apply Jensen’s inequality\footnote{Because $\varphi(x)=e^{-\Lambda x}$ is convex, thus $\mathbb E[\varphi(X)]\ge\varphi(\mathbb E[X])$.} to the convex function $x\mapsto e^{-\Lambda x}$ in one step, instead of handling three factors separately.
Applying Jensen's inequality,
\begin{equation}
   \mathbb{E}\!\bigl[e^{-\Lambda(\Delta+F)}\bigr]
   \;\ge\;
   e^{-\Lambda\,\mathbb{E}[\Delta+F]}
   =e^{-\Lambda\bigl(\mathbb{E}[\Delta]+1/\beta\bigr)},
\end{equation}
which plugged back into
$P_{\text{succ}}=(1-B(\rho,C))\,\mathbb{E}[e^{-\Lambda(\Delta+F)}]$
gives the lower bound
\begin{equation}
   l_{opt} = (1-B(\rho,C)) e^{\!\bigl[-\Lambda\bigl(\mathbb{E}[\Delta]+1/\beta\bigr)\bigr]}.
\end{equation}

Recall that $A$ is the \emph{stationary age} (residual life)
associated with the renewal sequence $\{Y_k\}$ of update gaps.
$\mathbb{E}[A] = \int_{0}^{\infty}\Pr\{Y>a\}~da = \frac{1}{\mathbb{E}[Y]} \int_{0}^{\infty}(1-F_Y(a)) a~da $. Following tail integration formula, $\mathbb{E}[A] = \frac{\mathbb{E}[Y^2]}{2\mathbb{E}[Y]}$.
Since $D\sim\mathrm{Exp}(\gamma)$,
we have 
\begin{equation}
   \mathbb{E}[\Delta]
   =\frac{1}{\gamma}
    +\frac{\mathbb{E}[Y^{2}]}
          {2\,\mathbb{E}[Y]},
\end{equation}
where $\mathbb{E}[Y]=T$,  $\mathbb{E}[Y^2]=T^2+2/\bar{r}^2$ (see Remark \ref{remark:EY_EY2}).
Finally, we have 
\begin{equation}
    l_{opt}=(1-B(\rho,C))e^{-\Lambda(\frac{1}{\gamma} + \frac{T}{2} + \frac{1}{T\bar{r}^2}+\frac{1}{\beta})}.
\end{equation}
\begin{remark}\label{remark:EY_EY2}
Because the $k$th and $(k-1)$th status packets are generated exactly one period $T$ apart and each experiences an independent exponential uplink delay, we have  
\begin{align}
    Y = S_k-S_{k-1} &= \bigl(kT+D_k\bigr)-\bigl((k-1)T+D_{k-1}\bigr) \nonumber \\ 
      &= T+(D_k-D_{k-1}),
\end{align}

where \(D_k,D_{k-1}\stackrel{\text{i.i.d.}}{\sim}\mathrm{Exp}(\gamma)\).
Independence and identical distribution imply  
\begin{equation}
    \mathbb{E}[D_k-D_{k-1}]
  =\mathbb{E}[D_k]-\mathbb{E}[D_{k-1}]
  =\frac{1}{\gamma}-\frac{1}{\gamma}=0.
\end{equation}

Consequently, 
\begin{equation}
\mathbb{E}[Y]
       =\mathbb{E}\!\bigl[T+(D_k-D_{k-1})\bigr]
       =T.
\end{equation}

In words, the mean inter-arrival interval observed at the AP remains equal to the deterministic generation period $T$; the symmetric, zero-mean jitter introduced by the two independent exponential delays does not change the average. 

By expanding $Y^{2}=(T+Z)^{2}$ and substituting the known variance $\mathrm{Var}[Z]=2/\gamma^{2}$ of the zero-mean Laplace difference $Z=D_k-D_{k-1}$, we directly obtain $\mathbb{E}[Y^{2}]=T^{2}+2/\gamma^{2}$.
\end{remark}

\section{Experiments}
\subsection{Experiment Setup}\label{sec:exp-design}

\begin{table}[htbp]
\centering\small
\caption{Controlled variables and their respective values (time unit: seconds)}
\label{tab:exp-vars}
\renewcommand{\arraystretch}{1.15}
\begin{tabular}{p{3cm}p{3.4cm}p{1cm}}
\toprule
Parameters & Values & Default \\ \midrule
User-arrival rate $\lambda_{\mathrm{in}}$
    & $\lambda_{\mathrm{in}} \in [5, 60]$ s$^{-1}$ & $40$ \\[2pt]
Service rate $\mu$ 
    & $\mu \in [20, 60] $ s$^{-1}$ & $30$ \\[2pt]
Number of threads $C$
    & $C \in \{1, 2, 3, 4, 5, 6\}$ & $2$ \\[2pt]
Update rate $\bar r$
    & $\bar r \in [5,100]$ s$^{-1}$ & $20$ \\[2pt]
Uplink rate $\gamma$
    & $\gamma \in [60, 200]$ s$^{-1}$ & $100$ \\[2pt]
Downlink rate $\beta$
    & $\beta \in [60, 200]$ s$^{-1}$ & $100$ \\ \bottomrule
\end{tabular}
\end{table}

\textbf{\textit{Comparison and Settings.}} To verify that the analytical bounds (i.e., the ideal upper bound $u_{\text{ideal}}$ and the operational lower bound $l_{\text{opt}}$) indeed enclose the \emph{true} task-success probability across a broad range of operating conditions, we employ a one-factor-at-a-time approach. Meanwhile, we also consider the closed-form $P_{\text{succ}}$ derived in Lemma \ref{lem:theoretical_p}. All experiments are implemented using the SimPy discrete-event simulation framework\footnote{https://simpy.readthedocs.io}. In each experiment, we vary a single controlled variable from \tablename~\ref{tab:exp-vars} across its specified range, while all other parameters are held at their default settings.

We follow a simple and intuitive ``safety frame'' guideline: $\bar r / \gamma < 1$ and $,\lambda / \beta < 1$. The first condition ensures that there is no more than one status packet in transit on the uplink at any time. The second condition guarantees that the average number of user tasks on the downlink does not exceed one.

In addition, since the default service rate is $\mu = 30$ and there are $C = 2$ threads, the maximum user-arrival rate is set to $\lambda_{\mathrm{in}}^{\max} = \mu \times C = 60$ to match the server’s capacity. Meanwhile, we set the minimum service rate to $\mu = 20$, so that the default server capability $20 \times C = 40$ can match the default user-arrival rate $\lambda_{\mathrm{in}} = 40$.

\textbf{\textit{Indicator.}} For each setting of the controlled variable, we run a simulation spanning $5{,}000$ time steps. Over this period, we record only two quantities: the number of tasks that successfully complete service on the server, denoted $N_{\text{succ}}$, and the total number of tasks arriving at the AP, denoted $N_{\text{arr}}$. The empirical task-success probability is thus calculated as
\begin{equation}
\hat P_{\mathrm{succ}}
=\frac{N_{\text{succ}}}{N_{\text{arr}}}.
\end{equation}

At the conclusion of each run, we calculate the analytical references $u_{\text{ideal}}$, $l_{\text{opt}}$, and the theoretical value $P_{\text{succ}}$ using the parameter settings specific to that experiment. For visual comparison, we plot all four curves (marking the empirical results $\hat P_{\text{succ}}$ with $\circ$) on a single figure as the controlled variable is varied.

\subsection{Simulation Results}
We begin by examining how variations in the user-task arrival rate $\lambda_{\mathrm{in}}$ (\figurename~\ref{fig:lamin-sweep}) and the service rate $\mu$ (\figurename~\ref{fig:mu-sweep}) affect the task–success probability.

\begin{figure}[htbp]
  \centering
  
  \begin{subfigure}[b]{0.49\linewidth}
    \centering
    \includegraphics[width=\linewidth]{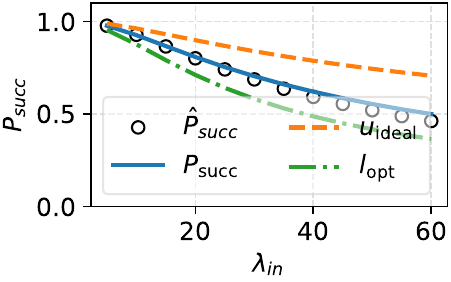}
    \caption{$P_{\mathrm{succ}}$ vs. 
             $\lambda_{\mathrm{in}}$}
    \label{fig:lamin-sweep}
  \end{subfigure}
  \begin{subfigure}[b]{0.49\linewidth}
    \centering
    \includegraphics[width=\linewidth]{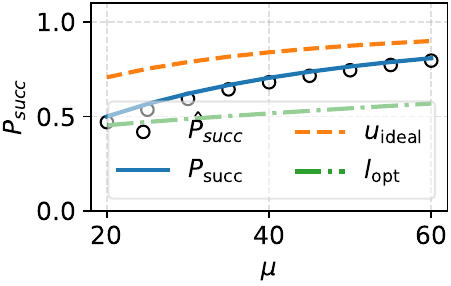}
    \caption{$P_{\mathrm{succ}}$ vs. 
             $\mu$}
    \label{fig:mu-sweep}
  \end{subfigure}
  \caption{Influence of user-task arrival rate $\lambda_{in}$ (left) and service rate $\mu$ (right) on task–success probability}
  \label{fig:lamin-mu}
\end{figure}

As shown in \figurename~\ref{fig:lamin-sweep}, the theoretical  task–success probability $P_{\mathrm{succ}}$, operational lower bound $l_{\mathrm{opt}}$, ideal upper bound $u_{\mathrm{idle}}$, and the empirical result $\hat{P}_{\mathrm{succ}}$ all decrease as the user-arrival rate $\lambda_{\mathrm{in}}$ increases, which aligns with intuition: Higher user access rates impose greater load on the server, resulting in a lower success probability. Notably, in the status-driven system, the theoretical expression for $P_{\mathrm{succ}}$ closely matches the experimental observations $\hat P_{\mathrm{succ}}$. The proposed bounds, $u_{\mathrm{idle}}$ and $l_{\mathrm{opt}}$, also effectively bracket the success probability as $\lambda_{\mathrm{in}}$ varies, demonstrating the validity of the theoretical framework.

Turning to \figurename~\ref{fig:mu-sweep}, we observe that as the service rate $\mu$ increases, the theoretical predictions again align closely with the empirical results. Overall, the two analytical bounds reliably enclose the real outcomes. Specifically, when the service rate $\mu$ is low, the lower bound provides a tight fit, with a gap of only about 1.6\% from the simulation results. As $\mu$ increases, both the theoretical value $P_{\mathrm{succ}}$ and the empirical result $\hat{P}_{\mathrm{succ}}$ rise, and are eventually constrained by the upper bound $u_{\mathrm{idle}}$, with a difference of nearly 1\%. In extreme scenarios, when $\mu$ is increased to $200$, the gap between the upper bound $u_{\mathrm{opt}}$ and the simulated result $\hat{P}_{\mathrm{succ}}$ remains within approximately 1.7\%.

\figurename ~\ref{fig:C-rbar} illustrates how the simulation results and theoretical bounds vary as $C$ and $\bar r$ are adjusted. Overall, the theoretical predictions consistently encompass the empirical results.

\begin{figure}[htbp]
  \centering
  \begin{subfigure}[b]{0.49\linewidth}
    \centering
    \includegraphics[width=\linewidth]{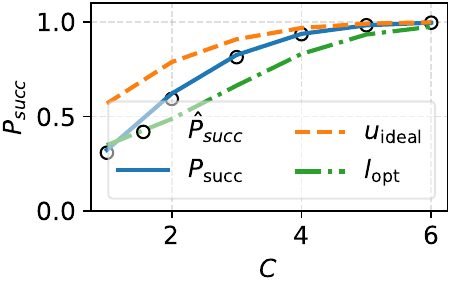}
    \caption{$P_{\mathrm{succ}}$ vs. $C$}
    \label{fig:C-sweep}
  \end{subfigure}
  \begin{subfigure}[b]{0.49\linewidth}
    \centering
    \includegraphics[width=\linewidth]{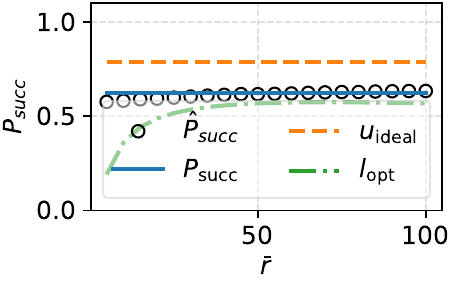}
    \caption{$P_{\mathrm{succ}}$ vs. $\bar r$}
    \label{fig:rbar-sweep}
  \end{subfigure}
  \caption{Influence of maximum number of threads $C$ (left) and status update frequency $\bar r$ (right) on task–success probability}
  \label{fig:C-rbar}
\end{figure}

As shown in \figurename~\ref{fig:C-sweep}, increasing the number of server threads $C$ continuously enhances the system’s task-handling capacity. When $C=6$, the theoretical bounds and the empirical result $\hat P_{\mathrm{succ}}$ all nearly approach 100\%, matching the theoretical value. This demonstrates that enlarging the thread pool effectively mitigates Erlang blocking. At this point, the system’s maximum processing capability, $C \times \mu = 6 \times 30 = 180$, is far greater than the default user arrival rate $\lambda_{\mathrm{in}} = 40$.

Turning to \figurename~\ref{fig:rbar-sweep}, it is noteworthy that the lower bound $l_{\mathrm{opt}}$ closely tracks the empirical results, with the smallest gap being only about 5\% (at $\bar r = 55$). Additionally, as the status update frequency $\bar r$ increases, the observed success probability shows only a modest improvement (about 5.8\%) and does not reach the theoretical maximum. This indicates that, in status-driven CFN systems, simply increasing the update frequency yields limited benefit; greater emphasis should be placed on enhancing the server’s processing capability. For smaller values of $\bar r$, we observe a larger gap between the lower bound and the measured results, which arises because a low update frequency can no longer accurately capture changes in the server state.

\figurename~\ref{fig:gamma-beta} illustrates the behavior of the analytical bounds and empirical results as the uplink and downlink rates ($\gamma$ and $\beta$) are varied. Overall, the system remains stable across these changes. This indicates that, when the number of status and task packets in transit is kept below one and both the server’s processing capacity and the user task arrival rate are stable, introducing additional transmission delay does not affect the distribution of the task-success probability—it merely increases the end-to-end latency for user tasks.

\begin{figure}[htbp]
  \centering
  \begin{subfigure}[b]{0.49\linewidth}
    \centering
    \includegraphics[width=\linewidth]{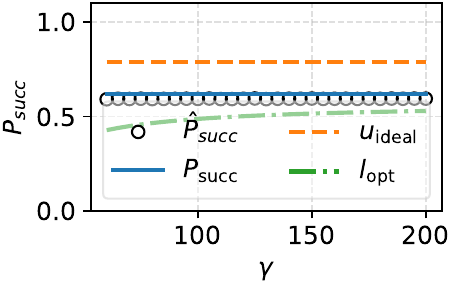}
    \caption{$P_{\mathrm{succ}}$ vs. 
             $\gamma$}
    \label{fig:gamma-sweep}
  \end{subfigure}
  \begin{subfigure}[b]{0.49\linewidth}
    \centering
    \includegraphics[width=\linewidth]{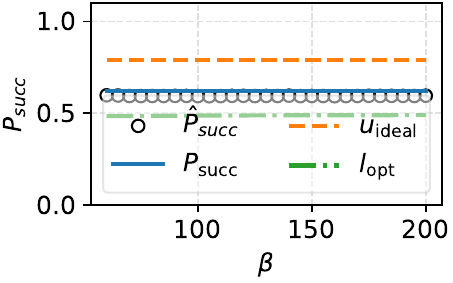}
    \caption{$P_{\mathrm{succ}}$ vs. $\beta$}
    \label{fig:beta-sweep}
  \end{subfigure}
  \caption{Influence of uplink $\gamma$ (left) and downlink $\beta$ (right) on task–success probability
}
  \label{fig:gamma-beta}
      \vspace{-12pt}
\end{figure}

Nevertheless, increasing the uplink and downlink rates can let us improve the frequency of status and task packets, thereby enhancing both the system’s computational throughput and the timeliness of status information. Combined with the findings from \figurename~\ref{fig:mu-sweep}, \figurename~\ref{fig:C-sweep}, and \figurename~\ref{fig:rbar-sweep}, this suggests that boosting link rates can further improve the overall task success probability.

\begin{figure}[htbp]
  \centering
  \begin{subfigure}[b]{0.49\linewidth}
    \centering
    \includegraphics[width=\linewidth]{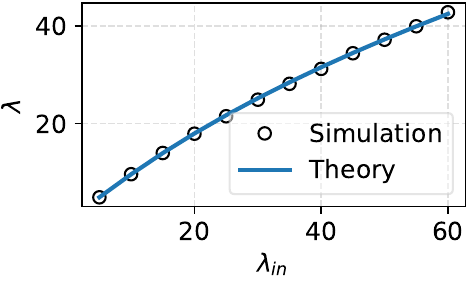}
    \caption{Forwarding rate $\lambda$}
    \label{fig:arrival-rate}
  \end{subfigure}
  \begin{subfigure}[b]{0.49\linewidth}
    \centering
    \includegraphics[width=\linewidth]{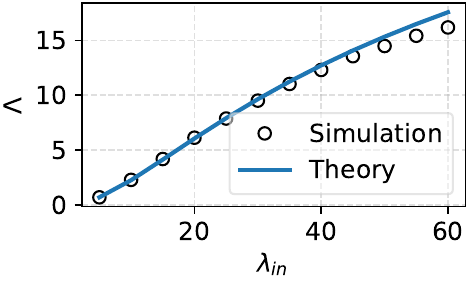}
    \caption{Resource hazard rate $\Lambda$}
    \label{fig:hazard-rate}
  \end{subfigure}
  \caption{Lemmas \ref{lem:forwarding-prob} and \ref{lem:hazard} verification via $\lambda$ and $\Lambda$}
  \label{fig:arrival-hazard}
      \vspace{-12pt}
\end{figure}

Finally, we examine how the analytical forwarding rate~$\lambda$
(Lemma~\ref{lem:forwarding-prob}) and the hazard rate~$\Lambda$
(Lemma~\ref{lem:hazard}) compare with simulation when the
exogenous arrival rate~$\lambda_{\mathrm{in}}$ is swept.
The results are plotted in Fig.~\ref{fig:arrival-hazard}.
The fixed–point solution for~$\lambda$ reproduces the measured data almost exactly across the entire range. The closed-form~$\Lambda$ also follows the trend, but begins to
slightly overestimate the measurements as $\lambda_{\mathrm{in}}$
increases. An explanation is
that when the user-arrival rate is high but status updates are slow, the AP always forwards a task \emph{after} it has detected an idle thread, so forwarded packets deliberately miss the most dangerous moments ($c(t)=1$); during the residual uplink delay the server often frees an additional thread, and some potential $1\!\rightarrow0$ events never happen—hence the measured hazard $\Lambda$ sits below the independence-based formula. In an additional test where we further increase the update rate to $\bar r=80$, the theory and measurement almost overlapped, the analytic prediction is $\Lambda\approx 17.6$ while the simulation produced $16.77$. This near-match confirms that a sufficiently high update rate removes the ``avoidance'' bias and brings the system behavior in line with the theoretical model.

\section{Discussion}
We tackle the joint process of task arrivals, periodic status refreshes, and independent link delays, and show that its behavior is fully captured by the theoretical value and two boundaries. The applicability of the proposed method can be explained as follows:
\begin{itemize}
    \item \textit{Forwarding probability.} 
    Whatever rule the AP uses, its whole effect can be written as   ``a task that reaches the AP is forwarded with probability   $P_{\text{fwd}}$ given the current number of busy threads.''   For the loss server studied here that number is simply the probability that at least one thread is still idle; if a small waiting buffer were added, $P_{\text{fwd}}$ would be the probability that at least one buffer slot is free. Once this probability is specified, the rest of the analysis remains unchanged.
    \item \textit{Wireless communication between users and the AP.} While this paper adopts a default setting where user tasks are assumed to arrive at the AP, possible wireless communication between users and the AP can be captured within our framework by introducing an additional random delay (modeled via its Laplace transform) into the holding window. Alternatively, the wireless transmission stage can also be viewed as part of the information ``waiting to be consumed'' period before the AP makes its forwarding decision.
    Further exploration of such extended modeling
    is left for future work.
    \item \textit{Random delays and Laplace transform.} 
    Every link or waiting stage enters the formula only through its Laplace image. Swapping the exponential link for a log-normal, or Weibull therefore means replacing one transform in $S(H)$ while other results remain intact.
\end{itemize}

Because the entire framework requires only a single probability and Laplace factors to characterize each new feature, it can be easily adapted to different server disciplines, link variations, or traffic patterns, without substantial effort.

\section{Conclusion}
We have introduced a unified analytical framework to characterize and bound the task–success probability in status‐driven CFN systems. By reducing the AP’s forwarding rule to a single forwarding probability and capturing all network and waiting delays through their Laplace transforms, we derive a closed‐form expression for the end‐to‐end success probability, together with upper and lower bounds. Our theoretical analysis rigorously accounts for stochastic task arrivals, multi‐thread server blocking, information staleness, and bidirectional link delays. Extensive simulations
demonstrate that the closed‐form expression and bounds enclose the empirical success rates within a few percent across diverse operating parameters. Moreover, our results reveal that, while increasing status‐update frequency yields diminishing returns, enhancing server processing capacity delivers substantial gains in success probability. 
Owing to its reliance on just two interchangeable inputs, our methodology can be readily extended to other dynamics.
We believe these insights will guide the design and optimization of status‐driven offloading strategies in next‐generation edge and CFN deployments. 
Future work focus on multi-server, multi-packet on the link, and other stochastic process related topics.
\clearpage
\balance
\bibliographystyle{plain}
\bibliography{references}
\end{document}